\documentclass[12pt]{article}
\usepackage{amssymb,amsmath,amsthm,amscd,latexsym}
\usepackage{mathrsfs}
\usepackage{mathrsfs}
\usepackage{amsfonts}
\usepackage{amsmath}
\usepackage{amssymb}
\usepackage{amscd}
\usepackage{mathrsfs}
\usepackage{amssymb}
\usepackage{amsmath}
\usepackage{amsthm}
\usepackage{color}
\usepackage{latexsym}
\usepackage{indentfirst}
\usepackage{enumitem}
\usepackage{anysize}
\usepackage{bbm}
\usepackage[normalem]{ulem}
\usepackage{soul}

\usepackage{cancel}

\renewcommand{\paragraph}{\roman{paragraph}}
\input cyracc.def

\newcommand{\Z}{\mathbb{Z}}
\newcommand{\F}{\mathbb{F}}

\newtheorem{theorem}{Theorem}
\newtheorem{corollary}{Corollary}

\newtheorem{lemma}{Lemma}
\newtheorem{proposition}{Proposition}

\theoremstyle{definition}
\newtheorem{definition}{Definition}
\newtheorem{remark}{Remark}

\newcommand{\RM}{\mathcal R \mathcal M}
\newcommand{\op}{\bigoplus}

\newcommand{\Fq}{\mathbb F_q}
\newcommand{\HH}{\mathcal{H}}

\begin{document}
\title{\bf How many weights can a cyclic code have ?
\thanks{This research is supported by National Natural Science Foundation of China (61672036), Excellent Youth Foundation of Natural Science Foundation of Anhui Province (1808085J20),
Technology Foundation for Selected Overseas Chinese Scholar, Ministry of Personnel of China (05015133).}}

\author{
\small{Minjia Shi$^{1}$, Xiaoxiao Li$^{1}$, Alessandro Neri$^{2}$\thanks{The author was supported by the Swiss National Science Foundation through grant no. 169510.},  and Patrick Sol\'e$^3$}\\ 
\and \small{${}^1$School of Mathematical Sciences, Anhui University, Hefei, 230601, China}\\
  \and \small{${}^2$Institute of Mathematics, University of Z\"{u}rich, Switzerland}
\and
\small{${}^3$CNRS/LAGA, University of Paris 8, 2 rue de la Libert\'e, 93 526 Saint-Denis, France}\\
}
\date{}
\maketitle
\begin{abstract}Upper and lower bounds on the largest number of weights in a cyclic code  of given length, dimension and alphabet are given. An application to irreducible cyclic codes is considered.
Sharper upper bounds are given for the special cyclic codes (called here strongly cyclic), {whose nonzero codewords have period equal to the length of the code}. Asymptotics are derived on the function
 $\Gamma(k,q),$ {that is defined as} the largest number of nonzero weights a cyclic code of dimension $k$ over $\F_q$ can have, and an algorithm to compute it is sketched. The nonzero weights in some infinite families of Reed-Muller codes, either binary or $q$-ary, as well as in the $q$-ary Hamming code are determined, two difficult results of independent interest.
\end{abstract}

\section{Introduction}
Recently, a series of works has been released on the study of the weight set of a code. In \cite{SZSC} Shi et al. conjectured that  the maximum number of nonzero weigths that a $k$-dimensional block code over $\F_q$ can have is $\frac{q^k-1}{q-1}$,  giving an answer only for the case $q=2$ or $k=2$. The complete proof of this conjecture was given in \cite{an18} where the two authors provided two different constructions, and independently in \cite{me18}. Shorter codes with the maximum number of weights for a given dimension were then discussed in \cite{a18,co18}. In the present paper we address the same type of questions for cyclic codes.
Thus, we study the function $\Gamma(k,q),$  {defined as} the largest number of nonzero weights a cyclic code of dimension $k$ over $\F_q$ can have.
We derive upper bounds on that quantity by simple counting arguments bearing on the cycle structure of the code. An alternative approach is to use weight concentration theorems, derived first in \cite{LN} in the language of linear recurrences. In the case of cyclic codes {whose nonzero codewords have period equal to the length} (called strongly cyclic codes in the sequel) we obtain smaller upper bounds {than the ones derived} for the class of all cyclic codes.
This suggests to study $\Gamma^0(k,q),$  {that is} the largest number of nonzero weights a strongly cyclic code of dimension $k$ over $\F_q$ can have.
This discrepancy in behaviour between $\Gamma(k,q),$ and $\Gamma^0(k,q),$ is particularly evident in the asymptotic upper bounds. We also derive lower bounds on these two functions, using special codes, or the covering radius of the dual code. An algorithm is given to compute $\Gamma(k,q),$ for small values of $kq.$  A first appendix derives the number of weights in several infinite families of $q$-ary Reed-Muller codes, including the extended Hamming code and some binary, ternary and quinary Reed-Muller codes. A second appendix derives the number of weights in the $q$-ary Hamming code when $q>2.$ These last two  technical results are of interest in their own right. While generating functions
or recursions \cite{ki07} are known in these cases, they are insufficient to determine the explicit weights in general.
 Exact enumeration of cyclic codes can require some deep techniques of Number Theory \cite{H}. In particular, the codes of Melas and Zetterberg give interesting lower bounds for a wide range of parameters. The use of the celebrated Delsarte bound on the covering radius of codes \cite{C+}, leads us to define a new combinatorial function ($T[n,k]$) of independent interest.

The material is organized as follows. The next section collects some background material on linear codes and cyclic codes. Section 3 is dedicated to upper bounds and Section 4 to lower bounds. Section 5 derives the asymptotic version of some of the preceding bounds. Section 6 gives an algorithm to compute $\Gamma(k,q).$ Section 7 concludes the paper and mentions some challenging open problems. Appendix I derives
the weight spectrum of some families of $q$-ary Reed-Muller codes. Appendix II determines the weights in the $q$-ary Hamming code for $q\ge 3.$
\section{Definitions and Notation}
\subsection{Linear codes}
A {\bf (linear) code} $C$ of length $n$  over a finite field $\mathbb{F}_{q}$ is a $\mathbb{F}_{q}$ vector subspace of  $\mathbb{F}_{q}^n.$ The dimension of the code is its dimension as a $\mathbb{F}_{q}$ vector space, and is denoted by $k.$ The elements of $C$ are called {\bf codewords}.

The {\bf dual} $C^ \bot$ of a code $C$ is understood {with respect to} the standard inner product.

 The {\bf (Hamming) weight} of ${\bf x} \in \F_q^N$ is the number of indices $i$ where $x_i\neq 0$, {and it is denoted by $w_H(\mathbf{x})$.}  The minimum nonzero weight $d$ of a linear code is called the {\bf minimum distance}. The {\bf dual distance} of a code is the minimum distance of its dual.
 Every linear code satisfies the {\bf Singleton bound} \cite[Th 11, Chap. 1]{MS} on its parameters
 $$d\le n-k+1.$$
 A code meeting that bound is called {\textbf{maximum distance separable (MDS)}}. See \cite[ Chap. 11]{MS} for general knowledge on this family of codes.
\subsection{Cyclic codes}
A {\bf cyclic} code of length $n$ over a finite field $\mathbb{F}_{q}$ is a $\mathbb{F}_{q}$ linear code of length $n$ invariant under the coordinate shift. Under the polynomial correspondence such a code can be regarded as an ideal in the ring $\mathbb{F}_{q}[x]/(x^n-1).$ It can be shown that this ideal is principal, with a unique monic generator $g(x),$ called {\bf the generator polynomial} of the code. The {\bf check polynomial} $h(x)$ is then defined as the
 quotient $(x^n-1)/g(x).$ A well-known fact is that the codewords satisfy a linear recurrence of characteristic polynomial the reciprocal polynomial of $h(x)$ \cite[p. 195]{MS}. Thus any codeword $\mathbf{c}$ can be continued by repetition into a bi-infinite periodic sequence $\widehat{\mathbf{c}}$ which is periodic of least period a divisor of $n.$
 The {\bf period} of a codeword $\mathbf{c}$ is understood to be the smallest integer $T_0$ such $\widehat{\mathbf{c}}_{i+T_0}=\widehat{\mathbf{c}}_i$ for all integers $i.$ Thus, the period is always a divisor of $n.$
 A cyclic code is {\bf irreducible} over $\F_q$ if its check polynomial $h(x)$ is irreducible over $\F_q[x].$
 The {\bf period} of a polynomial $h(x)\in \F_q[x],$ such that $h(0)\neq 0$ is the smallest integer $t$ such that $h(x)$ divides $x^T-1$ over $\F_q[x].$
 If $C$ is a cyclic code, its codewords are partitioned into orbits under the action of the shift. We call these orbits the {\bf cyclic classes} of $C.$
 \subsection{Combinatorial functions}
Here we introduce the combinatorial functions we are going to study in this work. Let $q$ be a prime power, $0\leq k \leq n$ be non negative integers. We define $\Gamma(k,q)$ as the largest number of nonzero weights of a cyclic code of dimension $k$ over $\F_q$. Moreover,  we define  $\Gamma(n,k,q)$ as the largest number of nonzero weights of a cyclic code of fixed  length $n$ and dimension $k$ over $\F_q,$ if such a code exists, and by zero otherwise.
 The same functions for strongly cyclic codes (to be defined below) are denoted by $\Gamma^0(k,q),$ and $\Gamma^0(n,k,q),$ respectively.
Recall also the functions $L(k,q)$ and $L(n,k,q)$ introduced in \cite{SZSC}, that represent respectively the maximum number of nonzero weights that a linear code of dimension $k$ over $\Fq$, and the  maximum number of nonzero weights for linear codes with a fixed length $n$. The function $L(k,q)$ was completely determined and shown to be equal to $\frac{q^k-1}{q-1}$ in \cite{an18}, while for $L(n,k,q)$ some partial answers were given in \cite{a18, SZSC}.
\section{Upper bounds}
\subsection{Cycle structure}
{Let $\rho:\F_q^n\longrightarrow \F_q^n$ denote the cyclic shift operator.
\begin{lemma}
 Let $C$ be an $[n,k]_q$ cyclic code and $\mathbf{c}\in C$ be a codeword of period $t$. Let moreover $\alpha \in \F_q^*$ and $i \in \{0,\ldots, n-1\}$ such that
$$\alpha \mathbf{c} =\rho^i(\mathbf{c}).$$
Then, $\alpha^r=1$, where $r=\frac{t}{\gcd(t,i)}$.
Moreover,  $\alpha$ belongs to the unique cyclic subgroup of $\F_q^*$ of order $\gcd(t,q-1)$.
\end{lemma}}\vspace{-0.7cm}
{\begin{proof}
 Let $\mathbf{c}$ be a codeword of period $t$ and suppose $\alpha \mathbf{c}=\rho^i(\mathbf{c})$. Then $\mathbf{c}=\rho^{ri}(\mathbf{c})=\alpha^r\mathbf{c}$, and this implies $\alpha^r=1$. Let now $H$ be the unique subgroup of $\F_q^*$ of order $\gcd(t,q-1)$. Since by definition $r|t$ and also $\alpha^{q-1}=1$, we get
$\alpha^{\gcd(t,q-1)}=1$, i.e. $\alpha\in H$.
\end{proof}}

If $C$ is a cyclic code, denote by $B_t$ the number of nonzero codewords of period $t$ it contains. A cyclic code such that $B_t=0$ for $1\le t<n$ shall be called {\bf strongly cyclic}.

{\begin{lemma}\label{lem:BT} If $C$ is an $[n,k]_q$ cyclic code with $s$ nonzero weights, then
$$s \le \sum_{t |n} \frac{B_t}{\mathrm{lcm}(t,q-1)}\le 1 + \sum_{1<t |n} \frac{B_t}{\mathrm{lcm}(t,q-1)}.$$
\end{lemma}}
 \vspace{-0.7cm}
{\begin{proof} The number of cyclic classes of codewords of period $t$ is at most $\frac{B_t}{t}.$ All codewords in the same class share the same weight.  Then, we use the Lemma above. Let $\mathbf{c}$ be a codeword and $H$ be the unique subgroup of $\F_q^*$ of order $\gcd(t,q-1)$. Now, for every representative $\alpha$ of $\F_q^*/H$, the codeword $\alpha \mathbf{c}$ gives a different class that obviously shares the same weight of the one of $\mathbf{c}$. Hence, there are at most $\frac{B_t \gcd(t,q-1)}{t (q-1)}=\frac{B_t}{\mathrm{lcm}(t,q-1)}$ distinct weights among these codewords, and this shows the first inequality. The second directly follows from the fact that if a codeword has period $1$ then it is necessarily a multiple of the vector of all ones, and therefore $B_1\in \{0,q-1\}$.
 \end{proof}}

\noindent{{\bf Example.} Consider the code of dimension $2$ over $\F_5,$ with length $20$ and check polynomial $x^2+x-1.$
 This code contains the Fibonacci numbers read mod $5$ \cite[A082116]{E}. It can be checked to contain $4$ codewords of period $4$ (namely $1,3,4,2, $ repeated five times) and $20$ codewords of period $20.$

 {This simple counting lemma has two important applications. First, we improve the upper bound on  $L(k,q)$ of \cite[Prop. 2]{SZSC} by a factor of at least $\frac{n}{q-1}$ for some large class of cyclic codes, up to a factor $n$ for another subclass.}
{{\theorem \label{brut} If $C$ is a $[n,k]_q$ strongly cyclic code with $s$ nonzero weights, then $$s \le \frac{q^k-1}{\mathrm{lcm}(q-1,n)}.$$ Thus $\Gamma^0(n,k,q)\le \frac{q^k-1}{\mathrm{lcm}(q-1,n)}.$ }}
{\begin{proof}
We apply Lemma \ref{lem:BT} when $B_t=0$ for $t<n,$ so that the sum in the right handside contains only one summand.
 \end{proof}}
{ One can observe that the function $\Gamma^0(n,k,q)$ has an intrisecly different behaviour than the one of $L(n,k,q)$. Indeed, it was proved in \cite{a18} $L(n,k,q)=\frac{q^k-1}{q-1}$ for $n\geq q^{3k-4}$, while  it is easy to see that a strongly cyclic code with that restriction on $n$ can not exist, i.e. $\Gamma^0(n,k,q)=0$ for $n\geq q^{3k-4}$ (it is actually already true for $n \geq q^k$ by Theorem \ref{brut}). This also implies that $\lim_{n\rightarrow \infty}\Gamma^0(n,k,q)=0$, and this is in contrast to the behaviour of  the function $L(n,k,q)$ that is non decreasing in $n$ and $\lim_{n\rightarrow \infty}L(n,k,q)=L(k,q)$,  }
 {As a second implication of Lemma \ref{lem:BT}, in the general case when many of the $B_t$'s are nonzero, we can prove the following result.}
{ {\theorem If $C$ is an $[n,k]_q$ cyclic code with $s$ nonzero weights,  then $$s\le 1+ (q^k-1)\bigg( \sum_{1<t |n} \frac{1}{\mathrm{lcm}(t,q-1)^2}\bigg)^{\frac{1}{2}}.$$}}\vspace{-0.7cm}
{\begin{proof}
   From the second  inequality  in  Lemma \ref{lem:BT}, we get
$$s-1\leq \sum_{1<t |n} \frac{B_t}{\mathrm{lcm}(t,q-1)}.$$
 Squaring both sides of this inequality and applying Cauchy-Schwarz inequality we obtain
  $$ (s-1)^2\le \bigg(\sum_{1<t |n} B_t^2\bigg)\bigg(  \sum_{1<t |n} \frac{1}{\mathrm{lcm}(t,q-1)^2}\bigg).$$
  By definition of the $B_t$'s note that $\sum_{t |n} B_t=q^k-1,$ implying $ \sum_{t |n} B_t^2 \le (q^k-1)^2.$ The result follows by taking the square root of both sides.
 \end{proof}}

 \subsection{Character sums}
 The following result can be derived by using the character sums techniques of \cite[Chapt. 8]{LN}.
 {\theorem \label{LN} If $C$ is an $[n,k]_q$ strongly cyclic code with $s$ nonzero weights, then
 $$ s\le 2 \left(1-\frac{1}{q}\right) q^{k/2}.$$
 Thus $$ \Gamma^0(n,k,q)\le 2 \left(1-\frac{1}{q}\right) q^{k/2}.$$}\vspace{-0.7cm}
\begin{proof}
 By \cite[Cor. 8.83]{LN} we know that the weights $w$ of $C$ lie in the range
 $$\left|n\left(1-\frac{1}{q}\right)-w \right|\le \left(1-\frac{1}{q}\right) q^{k/2}  .$$
 The result follows by computing the length of that interval.
 \end{proof}
 \subsection{Irreducible cyclic codes}

 The weight structure of irreducible cyclic codes has been a research topic since the first works of McEliece and others \cite{LN,M,D} due to their connection to Gauss sums and L-functions, and its intrinsic complexity.

 {\theorem \label{moi}If $C$ is an $[n=\frac{q^k-1}{N},k]_q$ irreducible cyclic code with a check polynomial of period $n,$ and $s$ nonzero weights, then $s\le N.$ }
 \begin{proof}
  Since the check polynomial $h(x)$ is irreducible, it generates the annihilating ideal of each sequence attached to a codeword. If the period of such a sequence were
  $T<n,$ then $h(x)$
  would divide $x^T-1,$ contradicting the hypothesis on the period of $h(x).$ Hence $C$ is strongly cyclic, and we can apply Theorem \ref{brut}. The result follows.
 \end{proof}

\noindent{\bf Example.} Consider the case of $N=2,$ and $q=p$ an odd prime. Such a code is well-known to be a two-weight code \cite{M}.

 A slightly sharper bound can be derived using the results in \cite{D}.

 {\theorem \label{1} If $C$ is an $[n=\frac{q^k-1}{N},k]_q$ irreducible cyclic code with a check polynomial of period $n,$ and $s$ nonzero weights then
 {$$s\le N_k=\gcd\left(N,\frac{q^k-1}{q-1}\right).$$ }}
\vspace{-0.7cm}
 \begin{proof}
  Follows by \cite[(12)]{D} which involves Gaussian periods of order $N_k.$
 \end{proof}

 This shows that Theorem \ref{moi}  can only be tight when $N=N_k,$ or, equivalently, $N$ divides $\frac{q^k-1}{q-1}.$
 Using Theorem \ref{LN}, another bound can be derived.

 {\theorem \label{2} If $C$ is an $[n=\frac{q^k-1}{N},k]_q$ irreducible cyclic code with a check polynomial of period $n,$ and $s$ nonzero weights then
 $s\le 2(1-\frac{1}{q})\sqrt{1+nN}.$ }

 \begin{proof} As explained in the proof of Theorem \ref{moi} we know that all nonzero codewords have period $n.$ Thus the code $C$ is strongly cyclic, and we can apply Theorem \ref{LN}.
 We get rid of $k$ in Theorem \ref{LN} by writing $q^k=1+nN.$
 \end{proof}

 \noindent{\bf Remark.} Depending on the relative values of $n$ and $N,$ either Theorem \ref{2}, or Theorem \ref{1} is sharper than the other.

 A slight improvement on Theorem \ref{2} can be derived for irreducible cyclic codes.

 {\theorem  If $C$ is an $[n=\frac{q^k-1}{N},k]_q$ irreducible cyclic code with a check polynomial of period $n,$ and $s$ nonzero weights then
 $$s\le 2\left(1-\frac{1}{q}\right)\left(\frac{n}{h}-\frac{1}{N}\right)\sqrt{1+nN},$$ where $h=\mathrm{lcm}(n,q-1).$}
 \begin{proof}
 The proof follows the lines of Theorem \ref{2} with \cite[Th. 8.84, (8.37)]{LN} replacing \cite[Cor. 8.83]{LN}. We get rid of $k$ by writing $q^k=1+nN.$
 \end{proof}
 \section{Lower bounds}
 \subsection{Special values}
 We begin with an easy bound.
 {\proposition For all prime powers $q,$ we have $\Gamma(k,q)\ge k.$}

 \begin{proof}
 The  {\bf universe code}, the cyclic $[k,k]_q$ code with generator the {constant polynomial $g(x)=1$}, has $k$ nonzero weights. This shows that $\Gamma(k,k,q)\ge k.$ The result follows by $\Gamma(k,k,q)\le \Gamma(k,q).$
 \end{proof}

 The following result is immediate by \cite{SZS2}. The proof is omitted.
 {\proposition For all prime powers $q,$ we have $\Gamma(2,q)=2.$}

 We recall now some classical cyclic codes. The {\bf repetition code} $R(n,q)$ is the ideal of $\F_q[x]/(x^n-1)$ with generator $\frac{x^n-1}{x-1}.$ Its dual is $P(n,q)=\langle (x-1) \rangle.$
 The {\bf binary Hamming code} $\mathcal{H}_m$ is the binary cyclic code of length $n=2^m-1$ with generator any primitive irreducible polynomial of $\F_2[x]$ of degree $m.$ Its dual the {\bf simplex code} $S_m$ is a one-weight code.
  {\theorem For all integers $n\ge 1$ and all prime powers $q$ with $(n,q)=1,$ we have that $\Gamma(n,1,q)=1,$ and that $\Gamma(n,n-1,q)$ is the number of nonzero weights in $P(n,q).$
 For all primes $m\ge 2,$  we have $\Gamma(n,n-m,2)=n-4, $ and $\Gamma(n,m,2)=1, $ where $n=2^m-1.$
 }
  \begin{proof}
  The first statement follows by the unicity of cyclic codes with dimension (resp. codimension) one. These are the repetition codes (resp. their duals). Their number of weights are easy to compute.
  To prove the second statement, observe that $x^{2^m}-x$ is the product of all monic irreducible polynomials whose degree divides $m$ \cite[Chap. 4, Th. 11]{MS}. If $m$ is a prime number, any divisor of $x^{2^m-1}-1$ of degree $m$ will have then to be irreducible.
  Thus, cyclic codes of dimension (resp. codimension) $m$ will have to be $S_m$ (resp. $\mathcal{H}_m$) or replicated versions of  $S_{m'}$ (resp. $\mathcal{H}_{m'}$) for $m$ a proper divisor of $m'.$
  The result follows on observing that the number of nonzero weights of $\mathcal{H}_m$ is $2^m-5$ (\cite[Chap. 6, Ex. (E2)]{MS}), and the fact that $S_m$  is a one-weight code \cite[Chap.1 \S 9 , Ex]{MS}.
   \end{proof}
   \noindent{\bf Remark.} The fact that the number of nonzero weights of $\mathcal{H}_m$ is $2^m-5$ can also be derived  applying part (1) of Theorem \ref{thm:bin} of the Appendix if we notice that $H_m$ is a punctured $R(m-2,m).$
   {\theorem Let $q>2$ be a prime power and $r\ge 2$ such that $\gcd(r,q-1)=1.$ We have the bound $$ \Gamma\left(\frac{q^r-1}{q-1},\frac{q^r-1}{q-1}-r,q\right)\ge \frac{q^r-1}{q-1}-2.$$}\vspace{-0.7cm}
   \begin{proof} Note that the $q$-ary Hamming code of length $\frac{q^r-1}{q-1}$ is {equivalent to a cyclic code} iff $\gcd(r,q-1)=1$ \cite{HP}.
   The result follows then by Theorem \ref{hamspec} of Appendix II.
   \end{proof}
The next result also uses Reed-Muller codes. We will need some generalization of the binomial coefficients given by the following generating function
$$ (1+z+\ldots+z^{q-1})^m=\sum_{\ell=0}^{m(q-1)}B(q,m,\ell)z^\ell.$$

{\theorem \label{thm:RMweights} With the notations above, we have the following results.\begin{eqnarray*}
\Gamma\bigg(2^m-1,\sum_{i=0}^{m-3}{{m} \choose i},2\bigg)&\ge& {2^{m-1}}-9\,\,\, \mbox{ for }\,\, m\ge 6,\\
\Gamma\bigg(3^m-1,\sum_{i=0}^{2m-2}B(3,m,i),3\bigg)&\ge& 3^m-3\,\,\, \mbox{ for }\,\, m\ge 1,\\
\Gamma\bigg(3^m-1,\sum_{i=0}^{2m-3}B(3,m,i),3\bigg)&\ge& 3^m-7\,\,\, \mbox{ for }\,\, m\ge 3,\\
\Gamma\bigg(5^m-1,\sum_{i=0}^{4m-3}B(5,m,i),5\bigg)&\ge & 5^m-4\,\,\, \mbox{ for }\,\, m\ge 1.\\
\Gamma\bigg(q^m-1,\sum_{i=0}^{(q-1)m-r-1}B(q,m,i),q\bigg)&\ge & q^m-r-{2}\,\,\, \mbox{ for }\,\, 0\le r\le \frac{q-{3}}{2}.
           \end{eqnarray*}
           }\vspace{-0.6cm}
           \begin{proof}
           The result combines  Appendix I with the fact that punctured Reed Muller codes are cyclic \cite[p.1312]{AK}. Note that the dimensions are not affected by puncturing due to the large minimum weights. The number of weights can decrease by at most $1$
           by transitivity of the automorphism group.
           \end{proof}
   The next two theorems rely on some deep algebraic geometric enumeration of cyclic codes \cite{LW,S,V2}. See \cite{H} for a survey.

   {\theorem For all integers $m\ge 3,$ we have $$\Gamma(2^m-1,2m,2)\ge \lceil 2^{m/2}\rceil, $$ and $$\Gamma(2^m+1,2m,2)\ge \lceil 2^{m/2}\rceil .$$}\vspace{-0.8cm}
   \begin{proof} The dual of the binary  Melas code is cyclic of parameters $[2^m-1, 2m].$ It is proved in \cite[Th. 6.3]{LW} that its nonzero weights are all the even integers $w$ in the range
   $$ \left|w-\frac{2^m-1}{2}\right|\le  2^{m/2}.$$
   Similarly, the dual of the Zetterberg code is an irreducible cyclic code of parameters $[2^m+1, 2m].$ It is proved in \cite[Th. 6.6]{LW} that its nonzero weights are all the even integers $w$ in the range
   $$ \left|w-\frac{2^m+1}{2}\right|\le  2^{m/2}.$$
   The result follows after elementary calculations.
   \end{proof}

   A ternary analogue is as follows.

   {\theorem For all integers $m\ge 2$ we have $\Gamma(3^m-1,2m,3)\ge \lceil 2\times 3^{\frac{m-2}{2}}\rceil. $ }
   \begin{proof}
   The dual of the ternary  Melas code is cyclic of parameters $[3^m-1, 3m].$ It is proved in \cite{V2}  that its nonzero weights are of the form $\frac{3^m-1+t}{3}$ with $t\in \Z,$ satisfying $t\equiv 1 \pmod{3},$ and $t^2 <3^m.$
   The result follows after elementary calculations.
   \end{proof}

   It is remarkable that the last two theorems imply lower bounds on $\Gamma(k,q),$
that are exponential in the dimension. It would be desirable to extend these results to $\Gamma(k,q)$ with $q$ a prime power $>3.$

 \subsection{Covering radius}
 Recall that the {\bf covering radius} $\rho(C)$ of a code $C$ is the smallest integer $t$ such that every point in $\F_q^n$ is at distance at most $t$ from some codeword of $C.$
 A combinatorial function that is, as far as we know, new, is $T[n,k,q],$ the largest covering radius of a cyclic code of length $n$ and dimension $k$ over $\F_q.$ Note that the closest classical function in that context is, for $q=2,$ the quantity
 $t[n,k],$ the smallest covering radius of a binary linear code of length $n$ and dimension $k$ \cite{C+}. Trivially $t[n,k]\le T[n,k,2].$
 The Delsarte bound \cite{MS}, stated for the dual of a linear code $C,$ is $\rho(C^\bot)\le s(C)$ \cite[Chap. 6, Th. 21]{MS}. With the above definitions, we can state the following result.

 {\proposition  \label{D} For all integers $n,k$ with $1\le k\le n,$ we have $$\Gamma(n,k,q)\ge T[n,n-k,q].$$}\vspace{-0.9cm}
\begin{proof}
 Upon using Delsarte bound for the dual of an $[n,k]_q$ code with $\Gamma(n,k,q)$ nonzero weights, which is, in particular, an $[n,n-k]_q$ code we see that $\Gamma(n,k,q)\ge T[n,n-k,q].$
 \end{proof}

\section{Asymptotics}
 To consider the number of weights of long codes of given rate, we study the behavior of $\gamma_q(R)$ defined for $0<R<1$ as
 $$\gamma_q(R)=\limsup_{n \rightarrow \infty}\frac{1}{n}\log_q(\Gamma(n,\lfloor Rn\rfloor,q)).$$
 {\theorem For all rates $R\in (0,1)$ we have $$ \gamma_q(R) \le R.$$ In particular, $\gamma_q(R)\le t(q),$ the unique solution in $(0,\frac{q-1}{q})$ of the equation $H_q(x)=x.$}

 \begin{proof}
 The  bound comes from the immediate inequalities $\Gamma(n,k,q)\le \Gamma(k,q)\le q^k-1.$
 \end{proof}

 Similarly for strongly cyclic codes we define
 $$\gamma^0_q(R)=\limsup_{n \rightarrow \infty}\frac{1}{n}\log_q(\Gamma^0(n,\lfloor Rn\rfloor,q)).$$
 We obtain a different upper bound.

 {\theorem For all rates $R\in (0,1)$ we have $$ \gamma^0_q(R) \le \frac{R}{2}.$$ In particular, $\gamma_q(R)\le t^0(q),$ the unique solution in $(0,\frac{q-1}{q})$ of the equation $H_q(x)=\frac{x}{2}.$}

 \begin{proof}
 The bound comes from the immediate inequalities $\Gamma^0(n,k,q)\le \Gamma^0(k,q)$ and Theorem \ref{LN}.
 \end{proof}
 \section{Numerics}

 We conjecture, but cannot prove, based on the figures of Table 1, that some local maxima of $n\mapsto \Gamma(n,k,2),$ for fixed $k$ of the form $k=\frac{s(s+1)}{2}$ are met for codes with check polynomials of the form $\prod_{i=1}^s h_i(x),$ where $h_i$ is irreducible of degree $i.$  Another motivation for the conjecture is that cyclic codes with irreducible check polynomials are one-weight codes in primitive length.
 \begin{center}
\small {Table 1: lower bounds on $\Gamma(k,q)$} \\
\end{center}

\begin{center}
\small
\begin{tabular}{|c|c|c|c|c|}
  \hline
   $\Gamma(k,q)\ge$  & $q$      & $k $    & $n$      & $h(x)$ \\
  \hline 8 & 2 & 6 & 21 & $(1+x)(1+x+x^2)(1+x+x^3)$ \\
  \hline 15 & 2 & 10 & 105 & $(1+x)(1+x+x^2)(1+x+x^3)(1+x+x^4)$\\
  \hline  11& 3& 6 &104&$ (x+1)(x^2+1)(x^3+2x+ 1)$ \\
  \hline 20&3& 10 & 1040& $(x+1)(x^2 +1)(x^3 + 2x + 1)(x^4+x+2)$\\
  \hline 11 &4& 6 & 315& $(x+1)(x^2 + x + w)(x^3+x+1)$\\
  \hline 18 &4& 8 & 315 & $(x+1)(x^2 + x + w)(x^2 + x + w^2)(x^3+x+1)$\\
  \hline
\end{tabular}
\end{center}
What can be noted from Table 1 is that the bound $\Gamma(k,q)\ge k$ of Proposition 1 is weak.

A systematic algorithm to compute $\Gamma(k,q)$ can be sketched as follows.
\begin{enumerate}
\item [(i)] Find all polynomials $h$ of degree $k$ of $\F_q[x],$ such that $h(0)\neq 0$;
\item [(ii)] For each $h$ compute its period $T_h$;
\item [(iii)] Count the number $s_h$ of nonzero weights of the cyclic code of length $T_h$ and check polynomial $h$;
\item [(iv)] Maximize $s_h$ over all $h$'s in Step (i).
\end{enumerate}
We justify this algorithm as follows.
\begin{proof}
View $h$ as the check polynomial of an $[n,k]$ cyclic code $C.$  If $h(x)=x^a H(x),$ with $a\ge 1,$ and $H(0)\neq 0,$ then $\gcd(x^n-1,h)=H(x)$ and $C$
has for dimension the degree of $H<k.$ Contradiction. Since $h(0)\neq 0,$ the period of $h$ is well-defined. All cyclic codes with check polynomial $h$ will have lengths a multiple of $T_h,$ and will be repetitions of the code of length $T_h,$
with the same number of weights.
\end{proof}
We illustrate this algorithm by the special case $k=3,q=2.$ The polynomials $h$ can take the following values
\begin{enumerate}
 \item[(i)] $x^3+x+1,\,x^3+x^2+1,$ when $T=7$ and $s=1$ (Simplex code)
 \item[(ii)] $x^3+1$ when $T=3$ and $s=3$ (Universe code)
 \item[(iii)] $x^3+x^2+x+1=\frac{x^4+1}{x+1}$ when $T=4$ and $s=2$ (Even weight code)
 \end{enumerate}

We conclude that $\Gamma(3,2)=3.$ More generally, we have the following results for $q=2$ and $q=3$. If the number $s_h$ of nonzero weight of the cyclic code of check polynomial $h(x)$ meets $\Gamma(k,q)$, we call $h(x)$ the optimal polynomial. We list some optimal polynomials $h(x)$ as follows. {The coefficients of the polynomial $h(x)$  are written in increasing powers of $x,$ for example for
$k=3,$ the entry $ 1001$ means $1+x^3.$}
\begin{center}
\small {Table 2: exact values of $\Gamma(k,2)$} \\
\end{center}
{
\begin{center}
\footnotesize
\begin{tabular}{|c|c|c|c|c|c|c|c|}
\hline
  & $k=3$ & $k=4$ & $k=5$ & $k=6$ & $k=7$ & $k=8$ & $k=9$  \\
  \hline
  $q=2$ & 3 & 4 & 5 & 8 & 9 & 12 & 16\\
  \hline
  $h(x)$ & (1001) & \begin{tabular}{c}

                         (11011) \\
                         (10001) \\

                       \end{tabular}
   & (100001) & (1110011) & (11100111) & (111101111) & \begin{tabular}{c}

                                                                                  (1100100001) \\
                                                                                  (1000010011) \\

                                                                                 \end{tabular}
   \\
   \hline
\end{tabular}
\end{center}}

{
\begin{center}
\small {Table 3: exact values of $\Gamma(k,3)$} \\
\end{center}
\begin{center}
\small
\begin{tabular}{|c|c|c|c|c|}
\hline
  & $k=3$ & $k=4$ & $k=5$ & $k=6$  \\
  \hline
  $q=3$ & 3 & 5 & 8 & 12 \\
  \hline
  $h(x)$ & \begin{tabular}{c}

              (2211) \\
              (2121) \\
              (2101)\\
              (2021) \\
               (2001)\\
               (1221)\\
               (1111)\\
               (1101)\\
               (1011)\\
               (1001)\\

            \end{tabular}
   & \begin{tabular}{c}

        (20221)\\
        (22011)\\
        (21021)\\
        (21011)\\
        (12111)\\
        (11121)\\
       (10201)\\
        (10101)\\

     \end{tabular}
    & \begin{tabular}{c}

        (222021)\\
        (211101) \\
         (210111)\\
         (202221)\\
         (121011)\\
         (112101)\\
         (110121)\\
         (101211)\\

      \end{tabular}
     & \begin{tabular}{c}

         (1201101)\\
          (1102101)\\
          (1012011)\\
          (1011021)\\

       \end{tabular}
      \\
   \hline
\end{tabular}
\end{center}}

\section{Conclusion and open problems}
In this paper, we have studied the largest number of distinct nonzero weights a cyclic code of given length and dimension could have. We have derived some upper bounds on that quantity that seem especially sharp for irreducible cyclic codes. Lower bounds appear weak so far, being mostly linear in $k,$ when the upper bounds are exponential. The Melas and Zetterberg codes provided lower bounds exponential in the dimension and it is worth extending these bounds to other range of parameters. The results on the weights of $q$-ary Reed-Muller and Hamming codes, while of interest in their own right only provide lower bounds that are polynomial in the dimension.


 So sharpening the lower bounds is the main open problem. Finding a pattern in the local maxima of $n \mapsto \Gamma(n,k,q)$ by running extensively the algorithm of the last section for large $n$'s might help. This programming effort could lead to a table of the function $\Gamma(k,q)$ for modest values of $kq,$ let us say $kq \le 100$ for instance.

 \noindent{\bf Acknowledgements.} The paper is dedicated to the memory of our friend and mentor G.D. Cohen (1951--2018). The authors wish to thank Professor C. Ding for helpful discussions.
\section{Appendix I:  Reed-Muller codes}
In both the appendices, the symbol $wt(C)$ denotes the set of weights of the code $C$, i.e.
$$wt(C)=\{i \in \mathbb N \mid \exists \, \mathbf{c} \in C \mbox{ s.t. } w_H(\mathbf{c})=i\}.$$
Moreover, let $t$ be a positive integer. Given $t$ sets of integers $A_1,\ldots, A_t$ we define the set of sums
$$\op_{i=1}^t A_i:=\left\{a_1+a_2+\ldots+a_t \mid a_i \in A_i\right\}. $$

Let $m$ be a positive integer and $\Fq$ be a finite field. Consider $R_m:=\Fq[x_1,\ldots, x_m]$ the ring of polynomials in $m$ variables over $\Fq$.
Moreover, list all the points of $\Fq^m$ as $\mathbf{P_1}, \ldots, \mathbf{P_n}$, where $n=q^m$, and consider the evaluation map
$$\begin{array}{rcl} \mathrm{ev}_m:  R_m & \longrightarrow & \Fq^n \\
f & \longmapsto & (f(\mathbf{P_1}), \ldots, f(\mathbf{P_n}) ). \end{array}$$

\begin{definition}
 Let $r,m$ be positive integers such that $0\leq r \leq (q-1)m$. The \textbf{$q$-ary Reed-Muller code of order $r$ in $m$ variables} is  defined by
$$\RM_q(r,m):=\left\{ \mathrm{ev}_m(f) \mid f\in R_m, \deg(f)\leq r\right\}.$$
\end{definition}

Observe that the choice of the order of the points $\mathbf{P_i}$ of $\Fq^m$ does not matter. Indeed, different choices lead to equivalent codes. However, it is possible to define the order in a smart way, such that we have an inductive construction of the $q$-ary Reed-Muller codes.
Let $\gamma$ be a primitive element of $\Fq$. For $m=1$ we consider $P_1=0$ and $P_i=\gamma^{i-1}$ for $i=2,\ldots,q$. Inductively, if $\{\mathbf{P_1},\ldots,\mathbf{P_n}\}$ is the set of points chosen for $n=q^m$, then for $m+1$ we choose $\{\mathbf{P_1'},\ldots, \mathbf{P_N'}\}$, where $N=q^{m+1}$, as follows.
$$\begin{array}{ll}
\mathbf{P_i'}=(\mathbf{P_i}, 0) &\mbox{ for } 1\leq i \leq n .\\
\mathbf{P_{nj+i}'}=(\bf{P_i}, \gamma^j) & \mbox{ for } 1\leq i \leq n, 1\leq j \leq q-1.
\end{array}$$

With this choice of the order of the points, one can prove the following result whose proof is omitted since it can be found in \cite{BN}.
For the rest of this section, for any integers $s,t$, we denote by $G_q(s,t)$  a generator matrix for the $q$-ary-Reed Muller code of order $s$ in $t$ variables.

\begin{proposition}\label{prop:recursiveRMmatrix}
 There exist $\lambda_{i,j}\in \Fq$ for $1\leq i < j \leq q$ such that
$$G_q(r,m+1)=\begin{pmatrix} G_q(r,m) & \lambda_{1,2}G_q(r,m) & \ldots & \lambda_{1,q}G_q(r,m) \\
0 & \ddots & & \vdots\\
\vdots &  & G_q(r-q+2,m) & \lambda_{q-1,q}G_q(r-q+2,m) \\
0&&0& G_q(r-q+1,m)\end{pmatrix}.$$
\end{proposition}

\begin{remark}
 Observe that, in the case $m=1$, the Reed-Muller code $\RM_q(r,1)$ is simply the extended Reed-Solomon code of length $q$ and dimension $r+1$. It is well-known that in this case $wt(\RM_q(r,1))=\{q-r,q-r+1,\ldots, q\}$ (see \cite[Theorem 6]{EGS}).
\end{remark}

Let now denote by $n_q(r,m)$, $k_q(r,m)$ and $d_q(r,m)$ respectively the length, dimension and minimum distance of the Reed-Muller code $R_q(r,m)$. The following result can be found in \cite{BN}, and explains more about the structure of this family of codes.

\begin{proposition}\label{prop:RMstructure}
 Let $r,m$ be integers such that $0\leq r \leq (q-1)m$. Then
\begin{enumerate}
\item [(1)] $n(r,m)=q^m$.
\item [(2)] $k(r,m)= \sum_{i=0}^r B(q,m,i) $, where $B(q,m,i)$ denotes the coefficient of $z^i$ in the polynomial $(1+z+\ldots+z^{q-1})^m$.
\item [(3)] $d(r,m)=(q-S)q^{m-1-Q}$, where $r=Q(q-1)+S$ with $0\leq S \leq q-2$.
 \item [(4)] $\RM_q(r,m)^\perp=\RM_q(m(q-1)-r-1,m)$, and therefore $G_q(r,m)$ is a parity check matrix for the code $\RM_q(m(q-1)-r-1,m)$.
\end{enumerate}
\end{proposition}

\begin{corollary}\label{cor:indRMweights}
 Let  $r,m $ be integers such that $m\geq 1$ and $0\leq r \leq m(q-1)$. Let $blue{\mathbf{u_1}}, \ldots, \mathbf{u_q} \in \RM_q((m-1)(q-1)-r-1,m-1)$. Then
$(\mathbf{u_1} \mid \ldots \mid \mathbf{u_q}) \in \RM_q(m(q-1)-r-1,m)$.
\end{corollary}

\begin{proof}
 Let $\mathbf{u_1}, \ldots, \mathbf{u_q} \in \RM_q((m-1)(q-1)-r-1,m-1)$. By Proposition \ref{prop:RMstructure} we know that  $\RM_q((m-1)(q-1)-r-1,m-1)^\perp=\RM_q(r,m-1)$, and therefore, $G_q(r,m-1)\mathbf{u_i}^T=0$ for $i=1,\ldots,q$. Moreover, for every $s\leq r$ we have $\RM_q(r,m-1)\supseteq \RM_q(s,m-1)$. This implies, by  part (4) of  Proposition \ref{prop:RMstructure}, that $ \RM_q((m-1)(q-1)-r-1,m-1)\subseteq \RM_q((m-1)(q-1)-s-1,m-1)$. Hence, we also have that $G_q(s,m-1)\mathbf{u_i}^T=0$.
Using the characterization given in Proposition \ref{prop:recursiveRMmatrix}, we get that $G_q(r,m)(\mathbf{u_1}\mid \ldots \mid \mathbf{u_q})^T=0$ and this completes the proof.
\end{proof}

As a direct consequence of Corollary \ref{cor:indRMweights}, we get the following result.

\begin{corollary}\label{cor:sumweightset} Let  $r,m $ be integers such that $m\geq 1$ and $0\leq r \leq m(q-1)$. Then
 $$wt(\RM_q(m(q-1)-r-1,m))\supseteq \bigoplus_{i=1}^q wt(\RM_q((m-1)(q-1)-r-1,m-1)).$$
\end{corollary}

We can finally give a general statement on the set of weights for some classes of $q$-ary Reed-Muller codes.

\begin{theorem}\label{thm:weightsRM}
For every positive integer $m$ and every integer $r$ such that   $0\leq r\leq \frac{q-3}{2}$, it holds that
$$wt(\RM_q(m(q-1)-r-1,m))=\{0,r+2,r+3,\ldots, q^m\}.$$
\end{theorem}

\begin{proof}
 We proceed by induction on $m$. For $m=1$ we get that $\RM_q(q-r-2,1)$ is the Reed-Solomon code of length $q$ and dimension $q-r-1$. Therefore, $wt(\RM_q(q-r-2,1))=\{r+2,r+3,\ldots, q\}$ by \cite[Theorem 6]{EGS}.

Suppose now it is true for $m-1$ and we want to prove the statement for $m$. Let $n=q^{m-1}$ and $N=q^m=qn$. We already know  by part (3) of Proposition \ref{prop:RMstructure}, that the minimum distance is equal to $r+2$. Moreover, by Corollary \ref{cor:sumweightset} and inductive hypothesis, we have that
\begin{align*}
wt(\RM_q(m(q-1)-r-1,m)) &\supseteq \bigoplus_{i=1}^q\{0,r+2,\ldots, n\}.
\end{align*}
  Let $x \in \left\{0,r+2,\dots, qn\right\}$, we need to prove that we can write $x=x_0+x_1+\dots+x_{q-1}$ with $x_i \in \{0,r+2,\ldots, n\}$. By Euclidean division, we have $x=an+b$, with $0 \leq a <q$ and $0 \leq b <n$. At this point we distinguish two cases. \\
\underline{Case 1: If $b \in \{0,r+2,\dots, n-1\}$}, then we choose $x_0=b$, $x_1=\dots =x_a=n$ and $x_{a+1}=\dots =x_{q-1}=0$.\\
\underline{Case 2: If $b \in\{1,\dots, r+1\}$}, then, necessarily $a \geq 1$. By hypothesis we have $n \geq q \geq 2r+3$. Therefore, $n-r-2+b \in \{0,r+2,\dots, n\}$ and we choose $x_0=r+2$, $x_1=n-r-2+b$, $x_2=\dots=x_a=n$ and $x_{a+1}=\dots=x_{q-1}=0$. This concludes the proof.
\end{proof}

\subsection{Binary Reed-Muller codes}
Here we provide an additional result not covered by Theorem \ref{thm:weightsRM} for binary Reed-Muller codes.

\begin{theorem}\label{thm:bin} Let $m$ be a positive integer, then we have
\begin{enumerate}
\item [(1)] If $m\geq 3,$ then
$wt(\RM_2(m-2,m))=\{0,2,4,\ldots,2^m\}\setminus\{2,2^m-2\}$.
\item [(2)] If $m\geq 6,$ then
$wt(\RM_2(m-3,m))\supseteq\{0,2,4,\ldots,2^m\}\setminus\{2,4,6,{10},2^m-2,2^m-4,2^m-6,{2^m-10}\}$.
\end{enumerate}
\end{theorem}

\begin{proof}
\begin{enumerate}
\item [(1)] We prove it by induction on $m.$ In the  case $m=3,$ a Magma \cite{Ma} computation shows that $wt(\RM_2(1,3))=\{0,2,4,6,8\}\setminus\{2,6\}$, $wt(\RM_2(2,4))=\{0,2,4,6,8,10,12,14,16\}\setminus\{2,14\}$ and $wt(\RM_2(3,5))=\{0,2,4,6,\ldots,32\}\setminus\{2,30\}$.
Now, suppose that the claim is true for $m\geq 3$, by Corollary \ref{cor:sumweightset}, we have $wt(\RM_2(m-1,m+1))\supseteq A$, where $A=\{i+j|i,j\in wt(\RM_2(m-2,m))\}$. It is easy to check that $A=\{0,2,4,\ldots,2^{m+1}\}\setminus\{2,2^{m+1}-2\}$. We know, by minimum distance arguments in Proposition \ref{prop:RMstructure}, that the integer $2$ is not in $wt(\RM_2(m-1,m+1))$. {Moreover, since the all ones vector is a codeword, also $2^m-2$ can not be a weight.} This completes the proof.

\item  [(2)] We also prove it by induction on $m$. In the case $m=6,$ a Magma \cite{Ma} computation shows that $wt(\RM_2(3,6))=\{0,2,4,6,8,\ldots,64\}\setminus\{2,4,6,10,54,58,60,62\}$.
Now, suppose that the claim is true for $m\geq 6$, then $B=\{i+j|i,j\in wt(\RM_2(m-3,m))\}$. It is also easy to check that $B=\{0,2,4,\ldots,2^{m+1}\}\setminus\{2,4,6,10,2^m-2,2^m-4,2^m-6,2^m-10\}$. We know, by minimum distance arguments in Proposition \ref{prop:RMstructure}, that the integers $2, 4$ and 6 are not in $wt(\RM_2(m-1,m+1))$. However, we don't know whether 10 belongs to $wt(\RM_2(m-1,m+1))$. Thus from Corollary \ref{cor:sumweightset}, we have $wt(\RM_2(m-1,m+1))\supseteq B$.
\end{enumerate}This completes the proof.
\end{proof}

\subsection{Ternary Reed-Muller codes}
In this subsection we give an additional result not covered by Theorem \ref{thm:weightsRM} for ternary Reed-Muller codes.

\begin{theorem}
Let $m$ be a positive integer, then we have
\begin{enumerate}
\item [(1)] If  $m\geq 1$, then $wt(\RM_3(2m-2,m))=\{0,3,4,5,\ldots,3^m\}$.

\item [(2)] If  $m\geq 3$, then $wt(\RM_3(2m-3,m))\supseteq \{0,6,8,9,\ldots,3^m\}$.
\end{enumerate}
\end{theorem}
\begin{proof}
\begin{enumerate}
\item [(1)] We prove it by induction on $m.$ The cases $m=1,2,3$ can be checked in Magma \cite{Ma} to be $wt(\RM_3(0,1))=\{0,3\}$, $wt(\RM_3(2,2))=\{0,3,4,5,6,7,8,9\}$ and $wt(\RM_3(4,3))=\{0,3,4,\ldots,27\}$.
Now, suppose that the claim is true for $m\geq 4$, from Corollary \ref{cor:sumweightset}, we have $wt(\RM_3(2m,m+1))\supseteq\ A=\{i+j+l|i,j,l\in wt(\RM_3(2m-2,m))\}$. It is easy to check $A=\{i+j+l|i,j,l\in wt(\RM_3(2m-2,m))\}=\{0,3,4,5,\ldots,3^{m+1}\}$. According to {part (3) of Proposition \ref{prop:RMstructure}}, $d(2m-2,m)=3$. This completes the proof.

\item [(2)] We also prove it by induction on $m$. The cases $m=3,4$ can be checked in Magma \cite{Ma} to be $wt(\RM_3(3,3))=\{0,6,8,9,\ldots,27\}$, $wt(\RM_3(5,4))=\{0,6,8,9,\ldots,81\}$.
Now, suppose that the claim is true for $m\geq 5$, then $B=\{i+j+l|i,j,l\in wt(\RM_3(2m-3,m))\}=\{0,6,8,9,\ldots,3^{m+1}\}$. However, we don't know whether 7 belongs to $wt(\RM_3(2m-1,m+1))$. Thus, from Corollary \ref{cor:sumweightset}, we have $wt(\RM_3(2m-1,m+1))\supseteq\{0,6,8,9,\ldots,3^{m+1}\}$. According to {part (3) of Proposition \ref{prop:RMstructure}}, $d(2m-3,m)=6$.  This completes the proof.

\end{enumerate}
\vspace{-0.7cm}
\end{proof}

\subsection{Reed-Muller codes over the field $\F_5$}
Here we provide an additional result not covered by Theorem \ref{thm:weightsRM} for  Reed-Muller codes over the field $\F_5$.

\begin{theorem}
Let $m$ be a positive integer, then we have
 $$wt(\RM_5(4m-3,m))=\{0,4,5,6,\ldots,5^m\} \mbox{ for } m\geq 1.$$
\end{theorem}

\begin{proof}

 We also prove it by induction on $m$. The cases $m=1,2$ can be checked with the aid of the software Magma \cite{Ma} to be $wt(\RM_5(1,1))=\{0,4,5\}$ and $wt(\RM_5(5,2))=\{0,4,5,6,\ldots,25\}$.
Now, suppose that the claim is true for $m\geq 3$, then $B=\{i+j+l+s+t|i,j,l,s,t\in wt(\RM_5(4m-3,m))\}=\{0,4,5,6,\ldots,5^{m+1}\}$. From Corollary \ref{cor:sumweightset}, we have $wt(\RM_5(4m+1,m+1))\supseteq\{0,4,5,6,\ldots,5^{m+1}\}$. Finally, from {part (3) of Proposition \ref{prop:RMstructure}}, we have $d(4m-3,m)=4$. This completes the proof.
\end{proof}

\section{Appendix II:  $q$-ary Hamming codes}
\subsection{Block Codes and Hamming Codes}

Let $q$ be a prime power, and let $\F_q$ denote the finite field with $q$ elements. We denote by $\F_q^*=\F_q\setminus\{0\}$ the multiplicative group of $\F_q$. It is well-known that $\F_q^*$ is a cyclic group, i.e. there always exists a generator $\gamma$ such that
$$\F_q=\{ \gamma^i \mid 1\leq i \leq q-1\} \cup \{0\}.$$

For $\mathbf{u}=(u_1,\ldots, u_n), \mathbf{v}=(v_1,\ldots, v_n) \in \F_q^n$ we denote by $\langle \mathbf{u},\mathbf{v}\rangle$ the standard inner product between $\mathbf{u}$ and $\mathbf{v}$, i.e.
$$\langle \mathbf{u},\mathbf{v}\rangle:=\sum_{i=1}^n u_iv_i.$$

Moreover, given two vectors $\mathbf{u} \in \F_q^{n_1}$, $\mathbf{v} \in \F_q^{n_2}$ we will use the notation $( \mathbf{u} \mid \mathbf{v})$ to indicate the vector in $\F_q^{n_1+n_2}$ obtained by concatenating  $u$ and $v$, i.e.
$$(\mathbf{u}\mid \mathbf{v})=(u_1,\ldots, u_{n_1}, v_1, \ldots, v_{n_2}),$$
where $\mathbf{u}=(u_1,\ldots, u_{n_1})$ and $\mathbf{v}=( v_1, \ldots, v_{n_2}).$
With this notation, given  $\mathbf{u_i},\mathbf{v_i} \in \F_q^{n_i}$ for $i=1,\ldots, m$, it trivally holds that
\begin{equation}\label{eq:sumweights}
\quad w_H(\mathbf{u_1} \mid \ldots \mid \mathbf{u_{m}})=\sum_{i=1}^m w_H(\mathbf{u_i}), \quad \mbox{ and }
\end{equation}
\begin{equation}
\langle (\mathbf{u_1} \mid \ldots \mid \mathbf{u_m}), (\mathbf{v_1} \mid \ldots \mid \mathbf{v_m})\rangle=\sum_{i=1}^m \langle \mathbf{u_i}, \mathbf{v_i} \rangle.
\end{equation}


 Suppose $r$ is a positive integer. On the set $\F_q^r\setminus\{\mathbf{0}\}$ we can consider the following equivalence relation:
$$\mathbf{u} \sim \mathbf{v} \Longleftrightarrow \exists \lambda \in \F_q^* \mbox{ such that } \mathbf{u}=\lambda \mathbf{v}.$$
The projective geometry of order $q$ and dimension $r$ is defined as
$$PG(r-1,q)= (\F_q^r \setminus \{\mathbf{0}\})/_\sim.$$
A set of representatives for $PG(r-1,q)$ is given by the set of all nonzero vectors of length $r$ whose first nonzero entries is equal to $1$. For this reason, the space $PG(r-1,q)$ can be embedded in $\F_q^r$. The image of this embedding will be denoted by $P_r(q)$. We will use this different notations in order not to confuse it with the space $PG(r-1,q)$, since we will often use the elements of $P_r(q)$ as elements of $\F_q^r$, where the notion of sum is well-defined. Since the cardinality of $P_r(q)$ is equal to the cardinality of $PG(r-1,q)$, we have that
$$\theta_q(r-1):=\frac{q^r-1}{q-1}=|P_r(q)|.$$
\begin{definition}\label{def:Hamming}
Let $q$ be a prime power, and let $r\geq 2$ be a positive integer. We define the \textbf{$q$-ary Hamming code} $\HH_r(q)$ as the $[\theta_q(r-1), \theta_q(r-1)-r]$ linear code over $\F_q$ whose parity check matrix $H_r$ is obtained by choosing all the vectors of $P_r(q)$ (without repetitions) as columns.
\end{definition}

Observe that in  Definition \ref{def:Hamming} the order of the choice of the columns does not really matter. Indeed, permuting the columns of the parity check matrix of a code gives rise to an equivalent code. All the parameters of a code, such that length, dimension and weight distribution are invariant under equivalence of codes. Therefore, for our purpose of studying the weight set of an Hamming code, we are free to choose the order of the columns of $H_r$ in the way we prefer by Definition \ref{def:Hamming}.

\begin{lemma}\label{lem:recursive PC}
 Let $H_{r-1}$ be the parity check matrix for the Hamming code  $\HH_{r-1}(q)$ as described in Definition \ref{def:Hamming}. Then we can obtain a parity check matrix for $\HH_r(q)$ as
$$H_r=\left(\begin{array}{ccc|ccc|ccc|c|ccc|c}
0 & \cdots &0 & 1 & \cdots & 1 & 1 & \cdots & 1 & \cdots & 1 & \cdots & 1 & 1 \\
\hline
&&&&&&&&&&&&&0\\
& H_{r-1} & & & \gamma  H_{r-1} & & & \gamma^2  H_{r-1} & & \cdots & & \gamma^{q-1} H_{r-1} & & \vdots \\
&&&&&&&&&&&&&0
\end{array} \right), $$
where $\gamma$ is a primitive element of $\F_q$ and $H_1=(1)$.
\end{lemma}

\subsection{The Main Theorem}
Hamming codes have been deeply studied in the theory of error correcting codes, since  they represent the only existing infinite family of perfect codes.  Indeed, in addition to Hamming codes,  the only perfect linear codes that are non-trivial, are the binary Golay code and the ternary Golay code \cite{HP,MS,TI}.

It is a well-known fact that $q$-ary Hamming codes always have minimum distance equal to 3. Moreover, a recursive formula for the weight distribution of a $q$-ary Hamming code $\HH_r(q)$ is known, and it can be found in \cite{ki07}. Anyway, from that recursion, it is not immediate how to deduce wether the number of codewords of weight $i$ in $\HH_r(q)$ is greater than zero or not. In the particular case that $r=2$ we know more about the structure of $\HH_2(q)$.

\begin{lemma}\label{lem:extRS}
For $r=2$ the Hamming code $\HH_2(q)$ is an MDS code of parameters $[q+1,q-1,3]$. In particular,
$wt(\HH_2(3))=\{0,3\}$ and $wt(\HH_2(q))=\{0,3,4,\ldots, q+1\}$ for every $q>3$.
\end{lemma}

\begin{proof}
The parameters of $\HH_2(q)$ are well-known to be $[q+1,q-1,3]$. From there we see that this code is MDS.
The  explicit determination of the weights follows then by \cite[Theorem 10]{EGS}.
\end{proof}
%

The above result represents, except for the case $q=3$, the base step of our proof for induction. Now, we give some auxiliary tools that will enable us to prove our main results in this section.

\begin{lemma}
 Let $q>2$ be a prime power, and let $r\geq 1$  be a positive integer. Then
$$\sum_{\mathbf{v} \in \F_q^r} \mathbf{v} = {\mathbf{0}}\ {\rm and }\ \sum_{\mathbf{u}\in P_r(q)} \mathbf{u}=\mathbf{{e}_r}, $$
where $\mathbf{{e}_r}=(0,0,\ldots, 0,1)^T$.
\end{lemma}

\begin{proof}
 For the first equality, we prove it by induction. Let $r=1$, $q\geq 3$,
and $\gamma$ be a primitive element of $\F_q$. Then
$$\sum_{v \in \F_q}v=\sum_{i=0}^{q-2}\gamma^i=\frac{\gamma^{q-1}-1}{\gamma-1},$$
since $\gamma\neq 1$. Moreover, $\gamma^{q-1}=1$ and then we get the result.

For $r>1$, we can write $\F_q^r=\bigcup_{\alpha \in \F_q}\mathcal T_{\alpha}$, where
$$ \mathcal T_\alpha:= \left\{(\alpha, \mathbf{v'}) \mid \mathbf{v'} \in \F_q^{r-1}\right\},$$
and the union is clearly disjoint. Therefore,
\begin{align*}
 \sum_{\mathbf{v} \in \F_q^r} \mathbf{v} =  \sum_{\alpha \in \F_q} \sum_{\mathbf{v'}\in \F_q^{r-1}} (\alpha, \mathbf{v'}) =  \sum_{\alpha \in \F_q} (q^{r-1}\alpha, \mathbf{0})  = \mathbf{0},
\end{align*}
where $q^{r-1}\alpha$ denotes (with a slight abuse of notation)
$$ \underbrace{\alpha + \ldots + \alpha}_{q^{r-1} \text{ times}},$$
and the second equality holds for inductive hypothesis.
\medskip

Concerning the second statement, we proceed again by induction. For $r=1$ the statement is clearly true, since $P_1(q)$ is the set with the only element ${e_1}=(1)$.

Suppose now $r>1$. Then we can write
$$P_r(q)=\left\{(0, \mathbf{u'}) \mid \mathbf{u'} \in P_{r-1}(q)\right\} \cup \left\{(1, \mathbf{v'}) \mid \mathbf{v'} \in \F_q^{r-1}\right\}. $$
Therefore,
$$ \sum_{\mathbf{u} \in P_r(q)} \mathbf{u} = \sum_{\mathbf{u'} \in P_{r-1}(q)} (0, \mathbf{u'}) + \sum_{\mathbf{v'} \in \F_q^{r-1}}(1, \mathbf{v'}) =\mathbf{{e}_r}+\mathbf{{0}}=\mathbf{{e}_r}, $$
where the second equality comes from inductive hypothesis and from the first part of this lemma.
\end{proof}

\begin{proposition}\label{prop:maxweight}
 For every $q>2$, and every integer $r\geq2$, there exists a codeword $\mathbf{c} \in \HH_r(q)$ of full weight, i.e., such that $w_H(\mathbf{c})=\theta_q(r-1)$, except for the case $(q,r)=(3,2)$.
\end{proposition}

\begin{proof}
 For $r=2$ the statement trivially follows from Lemma \ref{lem:extRS}. Suppose now $r\geq 3$, and
 let $\gamma$ be a primitive element of $\F_q$. Consider the set $$\mathcal D=\{\mathbf{w_0}, \mathbf{w_1},\ldots, \mathbf{w_{q-1}}\},$$
where $\mathbf{w_0}=\mathbf{{e}_1}=(1,0,\ldots,0)$ and $\mathbf{w_i}=(1,\gamma^i,\gamma^i,\ldots, \gamma^i)$ for $i=1,2,\ldots, q-2$ and $\mathbf{w_{q-1}}=(1,1,\ldots,1,0)$. It is easy to see that with this choice of $\mathcal D$, we have
$$\sum_{i=0}^{q-1}\mathbf{w_i} + \mathbf{{e}_r}=0.$$
Let $\alpha \in \F_q$. We consider the sum
\begin{align*}
\sum_{\substack{\mathbf{u} \in P_r(q) \\ \mathbf{u} \notin \mathcal D\cup \{e_r\}}} \mathbf{u}+(1+\alpha)\sum_{i=0}^{q-1} \mathbf{w_i}+\alpha \mathbf{{e}_r} &=\sum_{\substack{\mathbf{u} \in P_r(q) \\\mathbf{u}\neq {e}_r}} \mathbf{u} + \alpha \sum_{i=0}^{q-1} \mathbf{w_i} +\alpha \mathbf{{e}_r}\\
&=\alpha\left(\sum_{i=0}^{q-1} \mathbf{w_i} + \mathbf{{e}_r}\right)={\mathbf{0}}.
\end{align*}
Choosing $\alpha \in \F_q\setminus \{0,-1\}$, we have a linear combination of the columns of $H_r$ that gives the zero vector, and in which every column is multiplied by a nonzero coefficient. This gives a codeword that has all nonzero entries.
\end{proof}

 Proposition \ref{prop:maxweight}  proves the existence of  a codeword of maximum weight in every Hamming code, with the exception of $\HH_2(3)$. For all the remaining weights we will use an induction argument based on the following two  results.

\begin{lemma}\label{lem:sum}
 If $\mathbf{u_i}\in \HH_{r-1}(q)$ for $i=0, 1, \ldots, q-1$, then  $$(\mathbf{u_0} \mid \mathbf{u_1} \mid \ldots \mid \mathbf{u_{q-1}}\mid \alpha ) \in \HH_r(q),$$ where
$$\alpha=-\sum_{i=1}^{q-1} \langle \mathbf{u_i}, \mathbf{e}\rangle,$$
and $\mathbf{e}$ denotes the vector with all entries equal to $1$.
\end{lemma}

\begin{proof}
According to Lemma \ref{lem:recursive PC}, it is straightforward to see that the vector $(\mathbf{u_0} \mid \mathbf{u_1} \mid \ldots \mid \mathbf{u_{q-1}}\mid \alpha )$ satisfies the check equations given by the rows $i=2,\ldots, r$ of the parity check matrix. Moreover, the condition on $\alpha$ ensures that it also satisfies the first check equation.
\end{proof}

Let $q$ be a prime power, $n$ be a positive integer. Consider the space
$$(\F_q^n)^q=\underbrace{\F_q^n \times \ldots \times \F_q^n}_{q \text{ times}}.$$
Let $S_q$ denote the symmetric group on the set $\{0,1,\ldots, q-1\}$, $id$ denote identical permutation and consider the action of the group $(\F_q^*)^q \rtimes S_q$ on  $(\F_q^n)^q$ defined as

$$\arraycolsep=2.9pt\def\arraystretch{1.7}\begin{array}{rccl}\varphi: &((\F_q^*)^q \rtimes S_q) \times (\F_q^n)^q & \longrightarrow & (\F_q^n)^q \\
&\left((\boldsymbol{\beta},\sigma), (\mathbf{u_0} \mid \ldots \mid \mathbf{u_{q-1}})\right) & \longmapsto & (\beta_0\mathbf{u_{\sigma(0)}} \mid \ldots \mid \beta_{q-1}\mathbf{u_{\sigma(q-1)}}).
\end{array}$$

It is straightforward to see that the same group acts on $(\HH_{r-1}(q))^q$. In order to simplify the notation, for
$(\boldsymbol{\beta},\sigma) \in ((\F_q^*)^q \rtimes S_q)$, $\mathbf{u} \in (\F_q^n)^q$, we will write $$(\boldsymbol{\beta},\sigma) \cdot \mathbf{u}:=\varphi((\boldsymbol{\beta},\sigma),\mathbf{u} ).$$

\begin{proposition}\label{prop:concatenation}
Let $q>2$ be a prime power and let $(\mathbf{u_0} \mid \ldots \mid \mathbf{u_{q-1}}) \in (\HH_{r-1}(q))^q$. Then there exists $(\boldsymbol{\beta}, \sigma) \in (\F_q^*)^q \rtimes S_q$ such that
$$(\beta_0\mathbf{u_{\sigma(0)}} \mid \ldots \mid \beta_{q-1}\mathbf{u_{\sigma(q-1)}} \mid 0) \in \HH_r(q).$$
\end{proposition}

\begin{proof}
 Let $\mathbf{u_0}, \mathbf{u_1},\ldots, \mathbf{u_{q-1}} \in \HH_{r-1}(q)$. By Lemma \ref{lem:sum} we have that
 $(\mathbf{u_0} \mid \mathbf{u_1} \mid \ldots \mid \mathbf{u_{q-1}}\mid \alpha ) \in \HH_r(q),$
where $\alpha=-\sum_{i=1}^{q-1} \langle \mathbf{u_i}, \mathbf{e}\rangle$. If $\alpha=0$ we are done. Therefore, suppose $\alpha \neq 0$, and consider the set
$$\mathcal T:=\left\{i \in \{1,\ldots, q-1\} \mid \langle \mathbf{u_i}, \mathbf{e} \rangle\neq 0 \right\}.$$
Since $\alpha \neq 0$, then $|\mathcal T| \geq 1$. We distinguish two cases.

\underline{Case 1: $|\mathcal T| \geq 2$}. Let $r,s \in \mathcal T$ be two distinct elements. Consider a generic vector $\boldsymbol{\beta}=(\beta_0, \ldots, \beta_{q-1})$  whose entries are $\beta_i=1$ for $i \neq r,s$, and consider the elment
$(\boldsymbol{\beta}, id)\cdot (\mathbf{u_0} \mid \ldots \mid \mathbf{u_{q-1}})$. By Lemma \ref{lem:sum}, the vector
$$\left((\boldsymbol{\beta}, id)\cdot (\mathbf{u_0} \mid \ldots \mid \mathbf{u_{q-1}})\mid \tilde{\alpha}\right),$$
belongs to $\HH_r(q)$, where
\begin{align*}
\tilde{\alpha} & = -\sum_{i=1}^{q-1}\langle \beta_i \mathbf{u_i}, \mathbf{e} \rangle \\
         &=- \sum_{i=1}^{q-1} \langle \mathbf{u_i},\mathbf{e} \rangle +(1-\beta_r)\langle \mathbf{u_r},\mathbf{e}\rangle +(1-\beta_s)\langle \mathbf{u_s}, \mathbf{e}\rangle \\
   &= \alpha + (1-\beta_r)\langle \mathbf{u_r},\mathbf{e}\rangle +(1-\beta_s)\langle \mathbf{u_s}, \mathbf{e}\rangle.
\end{align*}
Imposing the condition $\tilde{\alpha}=0$ we get
\begin{equation}\label{eq:betars}
 (1-\beta_s)=- \frac{\alpha + (1-\beta_r)\langle \mathbf{u_r},\mathbf{e}\rangle}{\langle \mathbf{u_s},\mathbf{e}\rangle}.
\end{equation}
Since $q>2$, there exists $(\beta_r,\beta_s)\in \F_q^*\times \F_q^*$ satisfying Equation (\ref{eq:betars}), and we can conlcude.

\underline{Case 2: $|\mathcal T| =1$}.  Let $r$ be the only element in $\mathcal T$. We divide this case into two subcases:
\begin{itemize}
\item[(i)] If $\langle \mathbf{u_0}, \mathbf{e}\rangle=0$, then consider the transposition $\sigma=(0,r) \in \mathcal S_q$. Consider the element
$$(\mathbf{e}, \sigma)\cdot(\mathbf{u_0} \mid \ldots \mid \mathbf{u_{q-1}})=(\mathbf{v_0}\mid \ldots |\mathbf{v_{q-1}}).$$
By Lemma \ref{lem:sum}, we have
$$(\mathbf{v_0} \mid \ldots \mid \mathbf{v_{q-1}} \mid \alpha) \in \HH_r(q),$$
where
$$\alpha =-\sum_{i=1}^{q-1} \langle \mathbf{v_i}, \mathbf{e} \rangle = - \sum_{\substack{i=0 \\ i \neq r}}^{q-1} \langle \mathbf{u_i}, \mathbf{e} \rangle= -\sum_{\substack{i=0 \\ i \neq r}}^{q-1} 0= 0,$$
 and we are done.
\item[(ii)] Suppose now that $\langle \mathbf{u_0}, \mathbf{e}\rangle\neq 0$. Then there exists $s \in \{1,2,\ldots, q-1\}$ such that $s \notin \mathcal T$. We now consider the transposition  $\sigma=(0,s)$ and the corresponding vector
$$(\mathbf{e}, \sigma)\cdot (\mathbf{u_0} \mid \ldots \mid \mathbf{u_{q-1}}).$$
For this new vector we have now $|\mathcal T| =2$ and we can conclude by Case 1.
\end{itemize}
This completes the proof.\end{proof}

As a direct consequence, we immediately get the following result.
\begin{corollary}\label{cor:sumset}
 Let $q>2$ be a prime power, and let $r\geq 3$ be a positive integer. Then
$$wt(\HH_r(q))\supseteq \op_{i=0}^{q-1}wt(\HH_{r-1}(q)).$$
\end{corollary}

We are now ready to prove the main result.

\begin{theorem}\label{hamspec}
Let $q>2$ be a prime power, and let $r\geq 2 $ be an integer. Then
 $$wt(\HH_r(q))=\{0,3,\ldots,N\},$$
where $N:=\theta_q(r-1)$, except when $(q,r)=(3,2)$, in which case $wt(\HH_2(3))=\{0,3\}$.
\end{theorem}

\begin{proof}
 We prove it by induction. The base step is divided in two cases\\
\underline{Case 1: \ $ (q,2)$ with $q>3$}. It follows from Lemma \ref{lem:extRS}.\\
\underline{Case 2: \ $ (q,r)=(3,3)$}. We computed it with the aid of the software Magma \cite{Ma}.

Now, suppose that the claim is true for $(q,r-1)$ and we want to prove it for $(q,r)$. By hypothesis we know that
$$wt(\HH_{r-1}(q))=\{0,3,\ldots, n\}$$
where $n=\theta_q(r-2)$, and observe that $N=qn+1$. By Corollary \ref{cor:sumset}, we know that
$$wt(\HH_r(q))\supseteq \op_{i=0}^{q-1}wt(\HH_{r-1}(q))=\op_{i=0}^{q-1} \left\{0,3,\ldots, n\right\}.$$
Moreover, since $(q,r-1) \neq (3,2)$, we have that $n\geq 5$. By Proposition \ref{prop:maxweight} we already know that $N \in wt(\HH_r(q))$. Therefore it is enough to show that
$$\op_{i=0}^{q-1} \left\{0,3,\dots, n\right\}=\left\{0,3,\dots, qn\right\}.$$
It is clear that $N-1=qn$ belongs to our set, since $n\in wt(\HH_{r-1}(q))$. Let $x \in \left\{0,3,\dots, qn-1\right\}$, we need to prove that we can write $x=x_0+x_1+\dots+x_{q-1}$ with $x_i \in \{0,3,\ldots, n\}$. By Euclidean division, we have $x=an+b$, with $0 \leq a <q$ and $0 \leq b <n$. At this point we distinguish two cases. \\
\underline{Case 1: \ $b \in \{0,3,\dots, n-1\}$}. We choose $x_0=b$, $x_1=\dots =x_a=n$ and $x_{a+1}=\dots =x_{q-1}=0$.\\
\underline{Case 2: \ $b \in\{1,2\}$}. Then, necessarily $a \geq 1$. Moreover,  we already observed that $n \geq 5$. Therefore, $n-3+b \in \{0,3,\dots, n\}$ and we choose $x_0=3$, $x_1=n-3+b$, $x_2=\dots=x_a=n$ and $x_{a+1}=\dots=x_{q-1}=0$. This concludes the proof.
\end{proof}


\begin{thebibliography}{99}
 \bibitem{a18} T. L. Alderson, A note on full weight spectrum codes, arXiv preprint arXiv:1807.11798, 2018.
 \bibitem{an18} T. L. Alderson and A. Neri, Maximum weight spectrum codes, Advances in Mathematics of Communications, to appear, 2018.
\bibitem{AK} E. Assmus and J. Key, Polynomial codes and finite geometries, in {\it Handbook of Coding Theory}, {\bf 2}:1269--1343, C. Huffman, V. Pless, eds, North-Holland, Amsterdam, 1998.
 \bibitem{BN} T. Blackmore and G. Norton, Matrix product codes over $\F_q,$ Applicable Algebra in
Engineering, Communication and Computing, {\bf 12}(6): 477--500, 2001.
\bibitem{Ma} W. Bosma, J. Cannon, and C. Playoust, The Magma algebra system. I. The user language, J. Symbolic Comput., {\bf 24} (1997), 235-265, 1997, {\tt http://magma.maths.usyd.edu.au/magma/}
 \bibitem{C+} G.D. Cohen, I. Honkala, S. Litsyn, and A. Lobstein, {\em Covering codes}, volume {\bf 54}, North-Holland, Amsterdam, 1997.
 \bibitem{co18} G. D. Cohen and L. Tolhuizen. Maximum weight spectrum codes with reduced length. arXiv
preprint arXiv:1806.05427, 2018.
\bibitem{D}C. Ding and J. Yang, Hamming weights of irreducible cyclic codes, Discrete Mathematics, {\bf 313}(4): 434--446, 2013.
\bibitem{EGS}F. Ezerman, M. Grassl, and P. Sol\'e, The weights in MDS codes, IEEE Trananctions on Information Theory, {\bf IT-57}(1): 392--396, 2011.
\bibitem{HP}W. C. Huffman and V. Pless, {\it Fundamentals of error correcting codes}, Cambridge University Press, 2003.
\bibitem{H} N. Hurt, Exponential Sums and Coding Theory: A Review, Acta Applicandae Mathematicae, {\bf 46}(1): 49--91, 1997.
\bibitem{ki07} D. S. Kim, Weight distributions of Hamming codes, arXiv preprint arXiv:0710.1467, 2007
\bibitem{KP} F.R. Kschischang and S. Pasupathy, Some ternary and quaternary codes and associated sphere packings, IEEE Transactions on Information Theory, {\bf  38}(2): 227--246, 1992.
\bibitem{LW}G. Lachaud and J. Wolfmann, The weights of the orthogonals of the extended quadratic binary Goppa codes, IEEE Transactions on Information Theory, {\bf 36}(3): 686--692, 1990.
\bibitem{LN}R. Lidl and H. Niederreiter, {\em Finite fields}, volume {\bf 20}, Cambridge University Press, 1997.
\bibitem{MS} F.J. MacWilliams and N.J.A. Sloane, {\em The theory of error correcting codes}, North Holland, Amsterdam, 1977.
\bibitem{M}R.J. McEliece, Irreducible Cyclic Codes and Gauss Sums, in {\it Combinatorics}, 185--202, 1975.
\bibitem{me18} A. Meneghetti. On linear codes and distinct weights. arXiv preprint arXiv:1804.04373, 2018.
\bibitem{S}R. Schoof, Families of curves and weight distributions of codes,  Bulletin of the American
Mathematical Society,  {\bf 32}(2): 171--183, 1995.
\bibitem{SZS2} M. Shi, Z. Zhang, and P. Sol\'e, Two-weight codes and second order recurrences, Chinese Journal of Electronics, to appear, 2018.
\bibitem{SZSC} M. Shi, H. Zhu, P. Sol\'e, and G.D. Cohen, How many weights can a linear code have?, Designs Codes and Cryptography, 1--9. Doi.org/10.1007/s10623-018-0488-z, 2018.
\bibitem{E}N.J. Sloane, et al.: The on-line encyclopedia of integer sequences, 2003, {\tt http://www.oeis.org}.
\bibitem{TI} A. Tiet\"av\"ainen. On the nonexistence of perfect codes over finite fields. SIAM Journal on
Applied Mathematics, {\bf 24}(1): 88-96, 1973.
\bibitem{V2} G. Van der Geer and M. van der Vlugt, Artin-Schreier curves and codes, Journal of Algebra, {\bf 139}(1): 256--272, 1991.
\end{thebibliography}
\end{document}